\documentclass[11pt]{article}
\usepackage{graphicx}
\usepackage{amsmath}
\usepackage{multirow}
\usepackage{amssymb}
\usepackage{amsthm}
\usepackage{natbib}
\usepackage{moreverb}
\usepackage[T1]{fontenc}
\usepackage{gensymb}
\usepackage{pdfpages}
\usepackage[figuresright]{rotating}
\usepackage{geometry}
\usepackage{bbm}
\usepackage{setspace}
\usepackage{csquotes}
\renewcommand{\P}{{\text{pr}}}
\newcommand{\E}{{\mathbb{E}}}

\newcommand{\bx}{{\mathbf{x}}}
\newcommand{\bZ}{{\mathbf{Z}}}
\newcommand{\bz}{{\mathbf{z}}}

\newcommand{\by}{{\mathbf{y}}}
\newcommand{\bone}{{\mathbf{e}}}

\newcommand{\btau}{{\boldsymbol\tau}}
\newcommand{\bmu}{{\boldsymbol\mu}}

\newcommand{\bzeta}{{\boldsymbol\zeta}}

\newcommand{\boeta}{{\boldsymbol\eta}}
\newcommand{\bbeta}{{\boldsymbol\beta}}
\newcommand{\bbtau}{{\bar{\boldsymbol\tau}}}
\newcommand{\bhtau}{{\hat{\boldsymbol\tau}}}

\newcommand{\bm}{{{m}}}

\newcommand{\var}{{\text{var}}}

\newcommand{\Z}{{\mathbf{Z}}}

\newcommand{\cF}{{\mathcal{F}}}
\newcommand{\cC}{{\mathcal{C}}}

\newcommand{\cZ}{{\mathcal{Z}}}

\newtheorem{theorem}{Theorem}
\newtheorem{proposition}{Proposition}
\newtheorem{corollary}{Corollary}
\newtheorem{lemma}{Lemma}

\theoremstyle{definition}
\newtheorem{condition}{Condition}

\theoremstyle{remark}

\title{On mitigating the analytical limitations of finely stratified experiments}
\author{Colin B. Fogarty \thanks{Operations Research and Statistics Group, MIT Sloan School of Management, Massachusetts Institute of Technology, Cambridge MA 02142 (e-mail: \texttt{cfogarty@mit.edu})}}
\date{}
\begin{document}
\maketitle
\begin{abstract}
While attractive from a theoretical perspective, finely stratified experiments such as paired designs suffer from certain analytical limitations not present in block-randomized experiments with multiple treated and control individuals in each block. In short, when using an appropriately weighted difference-in-means to estimated the sample average treatment effect, the traditional variance estimator in a paired experiment is conservative unless the pairwise average treatment effects are constant across pairs; however, in more coarsely stratified experiments, the corresponding variance estimator is unbiased if treatment effects are constant within blocks, even if they vary across blocks. Using insights from classical least squares theory, we present an improved variance estimator appropriate in finely stratified experiments. The variance estimator is still conservative in expectation for the true variance of the difference-in-means estimator, but is asymptotically no larger than the classical variance estimator under mild conditions. The improvements stem from the exploitation of effect modification, and thus the magnitude of the improvement depends upon on the extent to which effect heterogeneity can be explained by observed covariates. Aided by these estimators, a new test for the null hypothesis of a constant treatment effect is proposed. These findings extend to some, but not all, super-population models, depending on whether or not the covariates are viewed as fixed across samples in the super-population formulation under consideration.
\end{abstract}

\section{Introduction}\label{sec:intro}
\subsection{The analytical limitations of finely stratified experiment}

When considering competing experimental designs, both theoretical and practical concerns must be taken into account. While the advice stemming from theoretical derivations is often in harmony with advice addressing issues of implementation, discordant recommendations can be encountered in the literature. As an illustration, consider the choice of granularity of stratification in a randomized experiment as it pertains to the variance of the resulting difference-in-means estimator of the average treatment effect. \citet{imb11} demonstrates that when considering, \textit{ex ante}, whether one should use a completely randomized experiment or a block-randomized experiment, the classical difference-in-means estimator for the average treatment effect in block-randomized experiment has a variance which cannot be higher than that of the estimator from a completely randomized experiment; see also \citet{fis35, coc57, cox58} and \citet{gre04} among many. By the same logic, a given block can be further broken into substrata while not increasing the estimator's variance. This leads \citet{ima09} and \citet{imb11} to prefer paired experiments from a theoretical perspective. \citet{kal13} further notes that from a population perspective, if one believes the response functions under treatment and control are Lipchitz with respect to some distance metric $\delta(\bx_i, \bx_j)$, then optimal pair matching with respect to $\delta(\bx_i, \bx_j)$ minimizes the variance of the difference-in-means estimator.

Moving away from designs with \textit{a priori} fixed block sizes, \citet{hig16} present a new experimental design called ``threshold blocking" which produces stratifications wherein each block contains \textit{at least} some number, call it $k$, individuals in each treatment arm. Taking $k=1$ in a treatment-control experiment then yields a design that is more flexible than pairing. \citet{hig16} present a near-optimal threshold blocking algorithm when one takes minimizing the maximal within-block covariate discrepancy between any two individuals in the same block as the objective. For the classical treatment-control experiment, the optimal stratification is mix of pairs and triplets, as any feasible stratum with four or more individuals can broken down into substrata of sizes two or three without increasing covariate discrepancy. \citet{sav15} illustrates that this additional flexibility from allowing for both pairs and triplets can result in lower estimator variance than a paired design, much in the same way that variable ratio matching tends to outperform fixed ratio matching in observational studies \citep{han04}. 

We define a \textit{finely stratified} design as one where within each block, there is either exactly one treated individual or exactly one control individual; both paired studies and optimal stratifications returned by threshold blocking satisfy this definition. We contrast these with \textit{coarsely stratified} designs, wherein each block has at least two individuals in each treatment group. Of course in principle this experimental taxonomy is not exhaustive as a treatment-control experiment could have both fine and coarse strata; we ignore this possibility in what follows. The preceding discussion has illustrated the theoretical merits of fine stratifications relative to coarse stratifications; however, finely stratified designs face  certain ``analytical limitations" avoided by coarsely stratified designs \citep{kla97, imb11, sav15}. As is well known, the true variance of difference-in-means estimator for the sample average treatment effect cannot be identified without further assumptions being made on the individual level treatment effects. Following the tradition of \citet{ney23}, conventional estimators for this variance exist which are conservative in expectation with respect to the experimental design's randomization distribution; see \citet{gad01} for an overview. It is when considering the magnitude of conservativeness for different experimental designs' standard variance estimators that the practical issues faced by finely stratified designs come to light. 

As will be presented explicitly in \S \ref{sec:vcompare}, the conventional variance estimator for a paired experiment is conservative in expectation unless the average treatment effect is constant across pairs, in which case it is unbiased; however, the typical variance estimator in a coarsely stratified experiment is unbiased so long as the treatment effect is constant \textit{within} blocks, even if the effects are heterogeneous across blocks. The practitioner must conduct hypothesis tests and form confidence intervals for the sample average treatment effect using a variance estimator appropriate for the design at hand. Hence, if the practitioner believes that the blocks in her experiment were formed on the basis of effect modifying covariates, any benefits in precision from employing a finely stratified design may be washed away by the increased conservativeness of the corresponding variance estimator. \citet{kla97} write that ``these limitations lead us...to favour stratified designs in which there are at least two [units] in each stratum" \citep[p. 1753]{kla97}. \citet{imb11} similarly notes that ``[These limitations are an] important reason to prefer experiments with at least two units of each treatment type in each stratum" \citep[p. 17]{imb11}.

\subsection{An insight from classical least squares squares theory}
The analytical limitations of finely stratified experiments thus present an unappealing gap between theory and practice.  Practical limitations hinder the actualization of theoretical benefits, an issue which we now seek to mitigate. Recent work by \citet{aro13, lin13, fog16, blo16} and \citet{lu16} among others has shown how regression adjustment can be utilized to provide improved estimators for the average treatment effect in various experimental designs. In this work, we will demonstrate how illustrate how regression adjustment can be utilized to yield improved \textit{variance} estimators in finely stratified experiments while using the classical difference-in-means estimator for the average treatment effect, hence preserving the so-called ``hands above the table" analysis \citep{fre08, lin13}. The key takeaway from this work is that effect modification can be exploited in a finely stratified experiment to yield improved variance estimates even when the model is misspecified. As the potential impact of effect modification is the source of the discrepancy between the variance estimators in finely and coarsely stratified experiments, this serves to close the gap between variance estimators in these respective designs. See \citet{aba08, din16, aba17} for recent work on the role of effect modification in variance estimation in related contexts.

Before proceeding, let us take a detour into classical least squares theory to provide insight into the improvements which will follow.  Suppose we have $n$ responses $\mathbf{{y}} = (y_1,...,y_n)^T$, and an $n\times K$ centered matrix of covariates $\tilde{X} = (I - \mathbf{e}\mathbf{e}^T/n)X$, where $I$ is the identity matrix and $\mathbf{e}$ is a vector containing $n$ ones. Consider running two regressions, the first a regression of ${\by}$ on $\bone$ and the second a regression of ${\by}$ on ${\bone}$ and $\tilde{X}$. By orthogonality, the coefficient on the intercept column, $\hat{\beta}_0$, will equal the sample mean $\bar{y}$ in both regressions. On the other hand, the variance estimators for  $\hat{\beta}_0$ will differ between the two regressions. For the regression on the intercept, the classical variance estimator for $\hat{\beta}_0$ is $\var(\hat{\beta}_0) = \sum_{i=1}^n(y_i - \bar{y})^2/(n(n-1))$. For a regression of $y$ on $\bone$ and $\tilde{X}$, the classical variance for $\hat{\beta}_0$ is
$\var(\hat{\beta}_0\mid \tilde{X}) = \sum_{i=1}^n(y_i - \bar{y} - \mathbf{\tilde{x}}^T_i(\tilde{X}^T\tilde{X})^{-1}\tilde{X}^Ty))^2/(n(n-K-1)).$ As a result, $\var(\hat{\beta}_0\mid \tilde{X}) \lessapprox \var(\hat{\beta}_0)$. The use of this improved variance estimator, $\var(\hat{\beta}_0\mid \tilde{X})$, is typically justified by an ancillarity argument: if the assumptions underpinning the regression model are satisfied, then the distribution of $X$ is ancillary for inference on any slope coefficient ${\beta}_k$. The conditionality principle would then support conditioning on $X$ in the inference that follows, hence restricting attention to the relevant subset of the sample space. 

\citet{buj14} provide an illuminating discussion not only of the classical arguments for conditioning on $X$, but also of the breakdown of these arguments in the presence on model misspecification. The fundamental issue is that when $X$ is itself considered to be random, $X$ is ancillary for inference on $\beta_k$ if and only if the model is correctly specified. The framework considered therein is one of a practitioner jointly sampling responses and covariates $iid$ from some target population, with the target of inference being the best linear approximation to the response function for this population. In the analysis of randomized experiments, a generative model of this nature is often implausible, as individuals within a given experiment need not constitute a representative sample. As such, inference is performed on local estimands such as the average treatment effect for the individuals in the experiment at hand, with the act of randomization itself provides the basis for inference for these estimands \citep{ney23, fis35, rub74, imb15}. For these local estimands, conditioning on the covariates for the individuals in the experiment is justified without an ancillarity, argument, as the estimands are themselves defined with respect to the sample at hand. As will be illustrated, variance estimators which utilize $X$ will furnish improvements in power while facilitating Neyman-style conservative inference for the sample average treatment effect.

\section{The sample average treatment effect}
\subsection{Notation for a block-randomized experiment}
\label{sec:notation}


There are $B$ independent blocks. The $i^{th}$ of $B$ blocks contains $n_i$ individuals, of whom $n_{1i}$ receive the treatment and $n_{0i}$ receive the control. There are $N = \sum_{i=1}^Bn_i$ total individuals in the study. Let $Z_{ij}$ be an indicator of whether or not the $j^{th}$ individual in block $i$ receives the treatment, such that $\sum_{j=1}^{n_{i}}Z_{ij} = n_{1i}$ and $\sum_{j=1}^{n_{i}}(1-Z_{ij}) = n_{0i}$. A finely stratified experiment is then characterized by $\min\{n_{0i}, n_{1i}\} = 1$ for all $i$, while in a coarsely stratified experiment $\min\{n_{0i}, n_{1i}\} > 1$ for all $i$. Individual $j$ in block $i$ has a $K$-dimensional vector of measured covariates $\bx_{ij} = (x_{ij1},...,x_{ijK})$. Each individual has a potential outcome under treatment, $r_{1ij}$, and under control, $r_{0ij}$, $i=1,...,B; j = 1,...,n_{i}$. The pair of potential outcomes $(r_{1ij}, r_{0ij})$ is not jointly observable for any individual. Instead, we observe the response $R_{ij} = r_{1ij}Z_{ij} +r_{0ij}(1-Z_{ij})$ for each individual.  As a consequence, the individual level treatment effect  $\tau_{ij} = r_{1ij} - r_{0ij}$ is not observable for any individual, nor is the average of the treatment effects in any block $i$, $\bar{\tau}_i = n_i^{-1}\sum_{j=1}^{n_{i}}(r_{1ij} - r_{0ij})$ \citep{ney23, rub74}.

Let $\Omega$ be the set of $\prod_{i=1}^B \binom{n_i}{n_{1i}}$ possible values of $\mathbf{Z} = (Z_{11}, Z_{12},..., Z_{Bn_{B}})^T$ under the block-randomized design. Each $z \in \Omega$ has probability $|\Omega|^{-1}$ of being selected, where the notation $|A|$ denotes the cardinality of the set $a$. Let $\cZ$ denote the event $Z \in \Omega$. Quantities dependent on the assignment vector such as $\Z$ and $\mathbf{R} = (R_{11}, R_{12},...,R_{Bn_{B}})^T$ are random, whereas $\cF = \{(r_{1ij}, r_{0ij}, \bx_{ij}), i = 1,...,B, j = 1,...,n_{B}\}$ contains fixed quantities for the experiment at hand. In a block-randomized experiment, $\P(\bZ = \bz\mid\cF, \cZ) = \P(\bZ = \bz\mid\cZ)=  |\Omega|^{-1} = \left(\prod_{i=1}^B\binom{n_i}{n_{1i}}\right)^{-1}$, and $\P(Z_{ij} = 1 \mid \cF, \cZ) = \P(Z_{ij} = 1 \mid \cZ)= n_{1i}/n_{i}$.

\subsection{The estimand and the estimator}
\label{sec:ATE}
The sample average treatment effect, or $SATE$, is defined as
\begin{align*}\bar{\Delta} & = \frac{1}{N}\sum_{i=1}^B\sum_{j=1}^{n_{i}}\tau_{ij} = \frac{1}{B}\sum_{i=1}^B w_i\bar{\tau}_i,\end{align*} where $w_i = B(n_i/N)$. The conventional unbiased estimator for $\bar{\tau}_i$, the average treatment effect for individuals in block $i$, is simply the observed difference-in-means between the treated and control individuals in block $i$.
\begin{align*}
\hat{\tau}_i &= \sum_{j=1}^{n_{i}}\left(\frac{Z_{ij}r_{1ij}}{n_{1i}} - \frac{(1-Z_{ij})r_{0ij}}{n_{0i}}\right).
\end{align*}
The classical unbiased estimator for the overall sample average treatment effect $\bar{\Delta}$ is
\begin{align}\label{eq:dim} \hat{\Delta} &= B^{-1}\sum_{i=1}^Bw_i\hat{\tau}_i,\end{align} i.e. a weighted average of the block-specific estimators with $n_i/N$ serving as weights \citep[Chapter 2]{obs}.

\section{A comparison of standard variance estimators}\label{sec:vcompare}
\subsection{Conventional variance estimation in coarsely stratified experiments}
For block $i$, define the block-specific averages of the potential outcomes under treatment and control as $\bar{r}_{1i} = n_i^{-1}\sum_{i=1}^{n_{i}}r_{1ij}$ and $\bar{r}_{0i} = n_i^{-1}\sum_{i=1}^{n_{i}}r_{0ij}$. Further, define  $\sigma^2_{1i}$, $\sigma^2_{0i}$, and $\sigma^2_{\tau i}$ by
\begin{align*} \sigma^2_{1i} = \frac{1}{n_{i}-1}\sum_{j=1}^{n_{i}}\left(r_{1ij} - \bar{r}_{1i}\right)^2;\;\;\;\sigma^2_{0i} = \frac{1}{n_{i}-1}\sum_{j=1}^{n_{i}}\left(r_{0ij} - \bar{r}_{0i}\right)^2;\;\;\;\sigma^2_{\tau i} = \frac{1}{n_{i}-1}\sum_{j=1}^{n_{i}}\left(\tau_{ij} - \bar{\tau}_{i}\right)^2.\end{align*}
The variance of the sample average treatment effect estimator in block $i$, $\var(\hat{\tau}_i\mid \cF, \cZ)$, can be expressed as \citep[Theorem 6.2]{imb15}
\begin{align*}\var(\hat{\tau}_i\mid \cF, \cZ) &= \frac{\sigma^2_{1i}}{n_{1i}} + \frac{\sigma^2_{0i}}{n_{0i}} - \frac{\sigma^2_{\tau i}}{n_i}.\end{align*} This immediately yields the following expression for $\var(\hat{\Delta}\mid \cF, \cZ)$:
\begin{align*}\var(\hat{\Delta}\mid \cF, \cZ) &= \frac{1}{B^2}\sum_{i=1}^Bw_i^2\left(\frac{\sigma^2_{1i}}{n_{1i}} + \frac{\sigma^2_{0i}}{n_{0i}} - \frac{\sigma^2_{\tau i}}{n_i}\right).\end{align*}

This variance is unknown in practice because it depends on the missing potential outcomes. In a coarsely stratified experiment where we have $\min\{n_{1i}, n_{0i}\} \geq 2$ for all $i$, the conventional estimator for $\var(\hat{\Delta} \mid\cF, \cZ)$ is based on an appropriately weighted sum of the sample variances of the treated and control responses in each block. Let $\bar{R}_{1i} = n_{1i}^{-1}\sum_{i=1}^{n_{i}}Z_{ij}r_{1ij}$ and $\bar{R}_{0i} = n_{0i}^{-1}\sum_{i=1}^{n_{i}}(1-Z_{ij})r_{0ij}$ be the observed averages of responses for the treated and control individuals in block $i$. Further, let $s^2_{1i}$ and $s^2_{0i}$ be the sample variances for the responses of the treated and control units in block $i$,
\begin{align*}
s^2_{1i} = \frac{1}{n_{1i} - 1}\sum_{j=1}^{n_{i}}Z_{ij}(r_{1ij} - \bar{R}_{1i})^2;\;\;\;s^2_{0i} = \frac{1}{n_{0i} - 1}\sum_{j=1}^{n_{i}}(1-Z_{ij})(r_{0ij} - \bar{R}_{0i})^2
\end{align*}
The classical variance estimator in a coarsely stratified experiment takes on the following form:
\begin{align*}S^2_{CS} &= \frac{1}{B^2}\sum_{i=1}^Bw_i^2\left(\frac{s^2_{1i}}{n_{1i}} + \frac{s^2_{0i}}{n_{0i}}\right).\end{align*}
A well known fact dating back to \citet{ney23} is that this estimator yields conservative inference for the sample average treatment effect, since
\begin{align}\label{eq:biascs}\E[S^2_{CS}\mid \cF, \cZ] - \var(\hat{\Delta}\mid \cF, \cZ) &= \frac{1}{B^2}\sum_{i=1}^Bw_i^2\sigma^2_{\tau i}.\end{align}Hence, the variance estimator $S^2_{CS}$ is an upper bound on $\var(\hat{\Delta}\mid \cF, \cZ)$ in expectation unless the treatment effect is constant within each block (i.e. if for each block $i$, $\tau_{ij} = \bar{\tau}_{i}$ for $j=1,...,n_i$). This thus enables Neyman-style conservative inference on $\bar{\Delta}$ to proceed using $S^2_{CS}$.

\subsection{Classical results on variance estimation in finely stratified experiments}

In a finely stratified experiment, at least one of $s^2_{1i}$ and $s^2_{0i}$ will be undefined as  $\min\{n_{1i}, n_{0i}\} = 1$. As a result, the estimator $S^2_{CS}$ cannot be employed. To the best of our knowledge there does not exist a ``classical" variance estimator for the general class of finely stratified experiments without making assumptions such as additivity of treatment effects or equal variance of potential outcomes \citep{obs, han04, sav15}. In the particular case of paired designs where $n_{1i} = n_{0i} = 1$ for all strata, the classical variance estimator is simply the sample variance of the observed paired differences divided by the number of pairs,
\begin{align}\label{eq:pairc}
S^2_{P}&= \frac{1}{B(B-1)}\sum_{i=1}^B(\hat{\tau}_i - \hat{\Delta})^2.\end{align} \citet{ima08} discusses inference for the sample average treatment effect within a paired design. Proposition 1 of that work illustrates that $S^2_P$ is also an upper bound in expectation for $\var(\hat{\Delta}\mid \cF, \cZ)$, and that the degree of the bias is given by
\begin{align}\label{eq:biasp}
\E[S^2_{P}\mid \cF, \cZ] - \var(\hat{\Delta}\mid \cF, \cZ) &= \frac{1}{B(B-1)}\sum_{i=1}^B(\bar{\tau}_i - \bar{\Delta})^2.\end{align}

A comparison of bias expressions (\ref{eq:biascs}) and (\ref{eq:biasp}) reveals the analytical limitations alluded to in \S 1.1. For a paired design, $S^2_P$ is biased upwards unless the average treatment effects are the same across pairs. In a coarsely stratified design, $S^2_{CS}$ is unbiased if there is additivity within blocks, even if there is effect heterogeneity across blocks. If the blocks were formed using covariates that are thought to be effect modifiers, it may be the case that the coarsely stratified design yields an unbiased estimator for the variance, while the paired design would yield a variance estimator that is substantially biased upwards. Were (\ref{eq:pairc}) the only variance estimator available to facilitate inference in a paired experiment, the practitioner in this case may well be justified in preferring the more coarsely stratified design as a means of shrinking confidence intervals and yielding more powerful hypothesis tests.

\section{Conservative variance estimators in finely stratified experiments}
\subsection{Two recipes with projection matrices}
Let $Q$ be an arbitrary $B \times L$ matrix with $L<B$, and let $H_Q = Q^T(Q^TQ)^{-1}Q$ be the orthogonal projection of $\mathbb{R}^B$ onto the column space of $Q$. Let $h_{Qij}$ be the $\{i,j\}$ element of $H_Q$. Define $y_i = \hat{\tau}_i/\sqrt{1-h_{Qii}}$ and $\mu_i = \bar{\tau}_i/\sqrt{1-h_{Qii}}$. Let $\mathbf{y} = (y_1,...,y_B)^T$, and let the analogous definitions hold for $\bmu$, $\bhtau$, and $\bbtau$. Finally, let $\Psi_Q$ be a $B\times B$ diagonal matrix whose $\{i,i\}$ entry equals $1/(1-h_{Qii})^2$

Let $W$ be a $B\times B$ diagonal matrix whose $i^{th}$ diagonal element contains $w_i = Bn_i/N$. We will now show that the matrix $Q$ can be used to produce two variance estimators which are conservative in expectation for $\var(\widehat{\Delta}\mid \cF, \cZ)$

Define the first of these estimators, $S^2_{1}(Q)$, as
\begin{align}\label{eq:sq}
{S}^2_{1}(Q)&=\frac{1}{B^2}\by^TW(I-H_Q)W.
\end{align} 
\begin{proposition}\label{prop:Q1}
If $Q$ is constant across all elements of $\Omega$:
\begin{align*}\E[S^2_{1}(Q) \mid \cF, \cZ] - \var(\hat{\Delta}\mid \cF, \cZ)
&=  \frac{1}{B^2}\bmu^TW(I-H_Q)W\bmu\geq 0
\end{align*}
\end{proposition}
\begin{proof} Define $\bmu$ as before, and let $\Lambda$ be the covariance matrix for $y$, a diagonal matrix with $\Lambda_{ii}= 1/(1-h_{Qii})\left({\sigma^2_{1i}}/{n_{1i}} + {\sigma^2_{0i}}/{n_{0i}} - {\sigma^2_{\tau i}}/{n_i}\right)$. Noting that $W(I-H_Q)W$ is symmetric,
\begin{align*}
B^2E[{S}^2_{1}(Q)\mid \cF, \cZ] &= tr(\Lambda W(I-H_Q)W) + \bmu^TW(I-H_Q)W\bmu\\
&= \sum_{i=1}^Bw_i^2\left(\frac{\sigma^2_{1i}}{n_{1i}} + \frac{\sigma^2_{0i}}{n_{0i}} - \frac{\sigma^2_{\tau i}}{n_i}\right) + \bmu^TW(I-H_Q)W\bmu
\end{align*}
Recalling that $\var(\hat{\Delta}\mid \cF, \cZ) = B^{-2}\sum_{i=1}^Bw_i^2\left({\sigma^2_{1i}}/{n_{1i}} + {\sigma^2_{0i}}/({n_{0i}}) - {\sigma^2_{\tau i}}/{n_i}\right)$
\begin{align*} \E\left[S^2_{1}(Q) \mid \cF, \cZ\right] - \var(\hat{\Delta} \mid \cF, \cZ) &= \frac{1}{B^2}\bmu^TW(I-H_Q)W\bmu \geq 0 ,\end{align*} where the last line stems from $(I-H_Q)$ being a projection matrix, and hence positive semi-definite.
\end{proof}

Define the second estimator, $S^2_{2}(Q)$, as
\begin{align}\label{eq:sq2}
{S}^2_{2}(Q)&=\frac{1}{B^2}\bhtau^TW(I-H_Q)\Psi_Q(I-H_Q)W\bhtau,
\end{align}

\begin{proposition}\label{prop:Q2}
If $Q$ is constant across all elements of $\Omega$:
\begin{align*}&\E[S^2_{2}(Q) \mid \cF, \cZ] - \var(\hat{\Delta}\mid \cF, \cZ)\\
&=  \frac{1}{B^2}\sum_{i=1}^Bw_i^2\left(\frac{\sigma^2_{1i}}{n_{1i}} + \frac{\sigma^2_{0i}}{n_{0i}} - \frac{\sigma^2_{\tau i}}{n_i}\right)\sum_{j\neq i}\frac{h_{Qij}^2}{(1-h_{Qjj})^2} + \frac{1}{B^2}\bbtau^TW(I-H_Q)\Psi_Q(I-H_Q)W\bbtau \geq 0
\end{align*}
\end{proposition}
\begin{proof} Define $\bbtau$ as before, and let $\Sigma$ be the covariance matrix for $\bhtau$, a diagonal matrix with $\Sigma_{ii}= 1/\left({\sigma^2_{1i}}/{n_{1i}} + {\sigma^2_{0i}}/{n_{0i}} - {\sigma^2_{\tau i}}/{n_i}\right)$. Noting that $W(I-H_Q)\Psi_Q(I-H_Q)W$ is symmetric,
\begin{align*}
B^2E[{S}^2_2(Q)\mid \cF, \cZ] &= tr(\Sigma W(I-H_Q)\Psi_Q(I-H_Q)W) + \bbtau^TW(I-H_Q)\Psi_Q(I-H_Q)W\bbtau.\end{align*}
The $\{i,i\}$ element of $\Sigma W(I-H_Q)\Psi_Q(I-H_Q)W$ is given by
\begin{align*}
(\Sigma W(I-H_Q)\Psi_Q(I-H_Q)W)_{ii} = w_i^2\left(\frac{\sigma^2_{1i}}{n_{1i}} + \frac{\sigma^2_{0i}}{n_{0i}} - \frac{\sigma^2_{\tau i}}{n_i}\right)\left(1 + \sum_{j\neq i}\frac{h_{Qij}^2}{(1-h_{Qjj})^2}\right)
\end{align*}
Recalling the form of $\var(\hat{\Delta}\mid \cF, \cZ)$ and noting that $(I-H_Q)\Psi_Q(I-H_Q)$ is positive semidefinite completes the proof.
\end{proof}

Propositions \ref{prop:Q1} and \ref{prop:Q2} illustrate that for any constant matrix $Q$ with $L<B$, the corresponding projection matrix can be utilized for conservative variance estimation in a finely stratified experiment through the estimators $S^2_1(Q)$ and $S^2_2(Q)$ defined in (\ref{eq:sq}) and (\ref{eq:sq2}). We will first illustrate that certain choices of $Q$ recover the standard variance estimator in a paired experiment when using $S^2_1(Q)$, and further suggest two conventional estimators for finely stratified experiments with varying block sizes. We will then show that the form of the bias expressions in Proposition 1 and 2 provides insight into choices for $Q$ which will provide improvements in variance estimation.
\subsection{Preliminary conservative variance estimators with equal and unequal block sizes}
Initially, let $\tilde{Q}_{1} = [\mathbf{e}, W\mathbf{e}-1]$ to be a $B\times 2$ matrix with a constant column along with a column corresponding to the centered weights (note that $B^{-1}\sum_{i=1}^Bw_i = 1)$. Define $Q_{1} = \tilde{Q}_{1} I_{2\times rank(\tilde{Q}_{1})}$, where $I_{k\times \ell}$ denotes a matrix of dimension $k\times \ell$ with ones on the diagonal and zeroes everywhere else; this removes the column $W\mathbf{e} - \mathbf{e}$ when block sizes are equal to avoid rank deficiency. We will now consider the implications of choosing $Q = Q_1$ in (\ref{eq:sq}) and (\ref{eq:sq2}) to define a conservative variance estimator.

When block sizes are equal $Q_1 = \mathbf{e}$, and hence the diagonal elements of the hat matrix associated with $Q_1$ equal $1/B$ for each observation. The variance estimator then takes on the simplified form
\begin{align*}
S^2_1(Q_{1})&= \frac{1}{B(B-1)}\sum_{i=1}^B(\hat{\tau}_i - \hat{\Delta})^2.\end{align*} In the case of matched pairs, this estimator is simply the sample variance of the observed paired differences divided by the number of pairs, hence recovering the classical variance estimator. Proposition 1 of \citet{ima08} for matched pairs can be viewed as a special case of our Proposition 1 with $Q = \mathbf{e}$. This also indicates that an additive treatment effect model implies unbiasedness of the estimator $S^2(Q_{1})$ for $\var(\hat{\Delta} \mid \cF, \cZ)$ in a finely stratified experiments with equal block sizes, even if the design is not paired. With equal block sizes, we have that $S^2_2(Q_{1})\geq S^2_1(Q_{1})$, meaning that the estimator $S_1^2(Q_{1})$ should always be preferred in this case.

With unequal block sizes, the $i^{th}$ diagonal elements of the hat matrix associated with $Q_{1}$ is $1/B + (w_i - 1)^2/\sum_{i=1}^B(w_i-1)^2$. Since the diagonal elements of the hat matrix depend on $w_i$, the estimator $S^2_{1}(Q_{1})$ will be a strict upper bound in expectation for $\var(\hat{\Delta} \mid \cF, \cZ)$ under an additive treatment effect model for finite samples $\bar{\tau}_i = 0$ for all $i$. So long as $(w_i-1)^2/\sum_{i=1}^B(w_i-1)^2\rightarrow 0$ for all $i$ as $B\rightarrow\infty$, the estimator $S^2_1(Q_{1})$ and $S^2_2(Q_{1})$ will both be asymptotically unbiased for $\var(\hat{\Delta} \mid \cF, \cZ)$ under an additive treatment effect (this condition would hold under the assumption that the block sizes are bounded, for example). In the unequal block case there is no longer a consistent ordering between $S_1^2(Q_{1})$ and $S_2^2(Q_{1})$, but the discrepancies tend to be minor: as will be demonstrated Theorem \ref{thm:2}, appropriately scaled versions of these two estimators converge in probability to the same limit under mild conditions.

\subsection{Improved variance estimation through exploiting effect modification}

For each block $i$, let $\mathbf{\bar{\bx}}_i$ be the vector of length $K$ whose $k^{th}$ entry is the average of the $k^{th}$ covariate for the individuals in block $i$, i.e. $\bar{x}_{ik} = n_{i}^{-1}\sum_{j=1}^{n_{i}}x_{ijk}$. Let $\bar{X}$ be the $B \times K$ matrix whose $k^{th}$ column contains $(\bar{x}_{1k}, \bar{x}_{2k},...,\bar{x}_{Bk})^T$ for $k=1,...,K$.  Let $M = (I-H_{Q_1})W\bar{X}$ be the weighted covariate means adjusted for $Q_{1}$. Let $Q_{2} = [Q_1, M]$. While the mutual orthogonality of $M$, $\bone$, and $W\bone-\bone$ within $Q_2$ is not required at this point, it facilitates forthcoming illustrations and makes clearer certain connections to heteroskedasticity consistent standard errors. Let $S^2_{1}(Q_{2})$ and $S^2_{2}(Q_{2})$ be the variance estimators corresponding to setting $Q = Q_{2}$ in (\ref{eq:sq}) and (\ref{eq:sq2}). 

To understand the potential benefits of the variance estimator $S^2_{1}(Q_2)$, note that from Proposition 1 the bias in $BS^2_{1}$ is $B^{-1}\mu^TW(I-H_{Q})W\mu$. Under mild regularity conditions described in \S \ref{sec:asy}, the diagonal elements of the hat matrix associated with $Q_{2}$ tend to 0 implying that $\mu_i \approx \bar{\tau}_i$ in sufficiently large samples. We can then think of $B^{-1}\mu^TW(I-H_{Q_{2}})W\mu$ as, approximately,the mean squared error from a regression of the weighted treatment effects, $W\bbtau$, on the weighted covariates, along with an intercept and a column for the block sizes. If the matrix $W\bar{X}$ contains covariates which are predictive of the treatment effects in different blocks, ${S}^2_{1}(Q_2)$ could yield a substantially less conservative estimator for $\var(\hat{\Delta} \mid \cF, \cZ)$ than the estimator $S^2_{1}(Q_1)$, which does not exploit potential effect modification.

For $S^2_{2}(Q_2)$, there is an additional connection to commonly employed standard error estimators in linear regression. In fact, since $Q_{2}$ was constructed such that $\bone$ is orthogonal to all other columns of $Q_2$, $S^2_{2}(Q_2)$ exactly corresponds to the square of the HC3 heteroskedasticity consistent standard error for the intercept column in a regression of $W\bhtau$ on $Q_{2}$ \citep{mac85, lon00}. The bias term for $BS^2_{2}(Q_2)$ is then approximately equal to $B$ times the HC3 variance for the intercept column of a regression of $W\bbtau$ on $Q_{2}$, which is itself a close approximation to the mean squared error from a regression of the weighted treatment effects $W\bbtau$ on $Q_{2}$.

Importantly, Propositons 1 and 2 make no assumption about the truth of the linear model generating the projection matrix $H_{Q}$. While the magnitude of the improvement from using ${S}^2_{\ell}(Q_2)$ instead of ${S}_{\ell}(Q_1)$ for $\ell=1,2$ depends on how well the weighted covariate means $W\bar{X}$ predict $W\bar{\btau}$, any choice of Q in (\ref{eq:sq}) or (\ref{eq:sq2}) will yield a variance estimator which is conservative in expectation for $\var(\hat{\Delta}\mid \cF, \cZ)$. As will now be shown, under mild conditions ${S}^2_{\ell}(Q_2)$  is asymptotically no worse than ${S}^2_{\ell}(Q_1)$ for $\ell=1,2$ regardless of the functional form describing the relationship between the observed covariates and the stratum-specific treatment effects. Further, both $B(S^2_1(Q_1) - S^2_2(Q_1))$ and $B(S^2_1(Q_2) - S^2_2(Q_2))$ converge in probability to zero.

\subsection{Asymptotic performance of variance estimators}\label{sec:asy}
We now give sufficient conditions which enable asymptotically valid inference for $\bar{\Delta}$ to proceed using $S^2_{\ell}(Q_1)$ and $S^2_{\ell}(Q_2)$ for $\ell=1,2$. In so doing, we will also quantify the potential improvements from exploiting effect modification through the variance estimator. The finite population asymptotics presented herein embed a given experiment with $B$ strata within an infinite sequence of experiments with increasingly many blocks. To reflect their changing values along this sequence, quantities such as $\bar{\Delta}$, $M$ and $W$ should be subscripted by $B$ for precision of notation; we omit this, trading precision for readability. Let $H_M = M(M^TM)^{-1}M^T$ be the hat matrix associated with $M$ as defined in the previous section, and consider the following regularity conditions.

\begin{condition}\textit{(Bounded Block Sizes)}
There exists a $C_1 < \infty$ such that $n_i < C_1$ for all $i$ and all $B$ as $B\rightarrow \infty$.
\end{condition}
\begin{condition}\textit{(Bounded Fourth Moments)}. \label{cond:1} There exists a $C_2 <\infty$ such that, for all $B$,\\ $B^{-1}\sum_{i=1}^B\sum_{j=1}^{n_i}w_i^4r_{1ij}^4/n_i < C_2$, $B^{-1}\sum_{i=1}w_i^4r_{0ij}^4/n_i < C_2$, $B^{-1}\sum_{i=1}w_i^4\bar{\tau}_i^4 < C_2$ and \\$B^{-1}\sum_{i=1}^B\sum_{j=1}^{n_i}w_i^4x_{ijk}^4/n_i < C_2$ for $k=1,..,K$. 

\end{condition}
\begin{condition}\textit{(Existence of Population Moments)}. \label{cond:2}  
\begin{itemize} \item $B^{-1}\sum_{i=1}^Bw_i\bar{\tau}_i$, $B^{-1}\sum_{i=1}^Bw_i^2\bar{\tau}_i$, $B^{-1}\sum_{i=1}^Bw_i^2\bar{\tau}_i^2$ and $B^{-1}\sum_{i=1}^Bw_i^2({\sigma^2_{1i}}/{n_{1i}} + {\sigma^2_{0i}}/{n_{0i}}$ $- {\sigma^2_{\tau i}}/{n_i})$ converge to finite limits as $B\rightarrow \infty$. 
\item $B^{-1}\sum_{i=1}^Bw_i \bar{\tau}_i{m}_{ik}$ converges to a finite limit for $k =1,...,K$ as $B\rightarrow \infty$. Let $\boeta_M$ be the vector of length $k$ containing these limits, i.e. $\eta_{Mk} = \lim_{B\rightarrow\infty}B^{-1}\sum_{i=1}^Bw_i \bar{\tau}_i{m}_{ik}$.
\item $B^{-1}M^TM$ converges to a finite, invertible matrix as $B\rightarrow \infty$. Call this limit $\Sigma_{M}$.
\end{itemize}
\end{condition}

Let $\bbeta_M = \Sigma_M^{-1}\boeta_M$. The following theorems illustrate that $S^2_{\ell}(Q_1)$ and $S^2_\ell(Q_2)$ for $\ell=1,2$ can all be used to conduct asymptotically conservative inference for the sample average treatment effect, $\bar{\Delta}$. After establishing asymptotic normality, we demonstrate that inference using $S^2_{\ell}(Q_2)$ will be no less powerful than that conducted using $S^2_{\ell}(Q_1)$ for $\ell=1,2$.
\begin{theorem}\label{thm:1}
Under Conditions 1-3 and conditional on $\cF$ and $\cZ$, 
\begin{align*}\sqrt{B}(\hat{\Delta} - \bar{\Delta}) &\overset{d}{\rightarrow}\mathcal{N}\left(0, B^{-1}\sum_{i=1}^Bw_i^2\left(\frac{\sigma^2_{1i}}{n_{1i}} + \frac{\sigma^2_{0i}}{n_{0i}} - \frac{\sigma^2_{\tau i}}{n_i}\right)\right).\end{align*}
\end{theorem}
\begin{theorem}\label{thm:2}
Under Conditions 1-3 and conditional on $\cF$ and $\cZ$, then for $\ell=1,2$,
\begin{align*}
BS^2_{\ell}(Q_1) - \var(\sqrt{B}\hat{\tau}\mid \cF, \cZ)&\overset{p}{\rightarrow}  \underset{B\rightarrow\infty}{\lim}  \frac{1}{B}\bar{\btau}^TW(I-H_{Q_{1}})W\bar{\btau};\\
BS^2_{\ell}(Q_2) - \var(\sqrt{B}\hat{\tau}\mid \cF, \cZ)&\overset{p}{\rightarrow}  \underset{B\rightarrow\infty}{\lim}  \frac{1}{B}\bar{\btau}^TW(I-H_{Q_{2}})W\bar{\btau}\\
&= \underset{B\rightarrow\infty}{\lim}  \frac{1}{B}\bar{\btau}^TW(I-H_{Q_{1}})W\bar{\btau} - \beta_M^T\Sigma_M\beta.\\
\end{align*}
\end{theorem}
\begin{corollary}
For $\ell=1,2$,
\begin{align*}
BS^2_{\ell}(Q_1) - BS^2_{\ell}(Q_2) &\overset{p}{\rightarrow} \bbeta_M^T\Sigma_M\bbeta_M \geq 0. \end{align*} 

\end{corollary}


The proofs are deferred to the appendix. The above results, in concert with Propositions 1 and 2, justify multiple means by which inference can be conducted for the sample average treatment effect, $\bar{\Delta}$, in finely stratified experiments. The results validate new standard error estimators for inference on the $SATE$ in finely stratified experiments while using classical weighted difference-in-mean estimator. Furthermore, these results highlight how effect modification can be leveraged to reduce the degree of conservativeness of the performed inference. As Corollary 1 demonstrates, standard errors derived by including suitably weighted average values for covariates within blocks are, asymptotically, never worse than those derived without including covariate information.

\section{Consonant and dissonant super-population formulations}
\subsection{Population-level causal estimands}
The preceding results make no assumptions about the manner by which individuals were selected for inclusion into the block-randomized experiment in the first place; that is, they neither require nor postulate the existence of a larger population from which individuals were drawn. The target of estimation, the sample average treatment effect, attests merely to the treatment effect for individuals in the sample at hand, and the act of randomization provides a reasoned basis for making probabilistic statements \citep{fis35}. That being said, it is sometimes desired to postulate that individuals in the study at hand were in fact draws from a super-population, and to perform inference on the average treatment effect within that super-population. 

\subsection{Conditional average treatment effect (CATE)}

As an initial super-population extension, suppose we consider the covariates $\bx_{ij}$ and the block sizes $\{n_1,...,n_B\}$ as fixed and consider the pairs of potential outcomes $(r_{1ij}, r_{0ij})$ as having arisen through the following sampling mechanism.
\begin{align*}
(r_{1ij}, r_{0ij}) &= (f_{1i}(\bx_{ij}), f_{0i}(\bx_{ij})) + (\epsilon_{1ij}, \epsilon_{0ij}),
\end{align*} where $(\epsilon_{1ij}, \epsilon_{0ij})$ are drawn from an arbitrary distribution with mean $(0,0)$ and block-specific variance-covariance matrix $\Sigma_{i\epsilon}$. Let $f_{1ij} = f_{1i}(\bx_{ij})$, and let $f_{0ij} = f_{0i}(\bx_{ij})$. Let $\mathcal{C} = \{\bx_{ij}\}$ be the set containing the covariates for all individuals. Within this super-population abstraction, the conditional average treatment effect, or $CATE$, in a finely stratified experiment is defined as
\begin{align}\label{eq:CATE}
\bar{\Delta}^{(C)} = \frac{1}{N}\sum_{i=1}^B\sum_{j=1}^{n_{i}}(f_{1ij} - f_{0ij})
\end{align}

Let $\bar{f}_i = n_i^{-1}\sum_{j=1}^{n_{i}}(f_{1ij} - f_{0ij})$, and let $\bar{\mathbf{f}} = (\bar{f}_1,..., \bar{f}_B)^T$. Note that (\ref{eq:CATE}) reflects the view of the covariates as fixed, in much the same way that conventional least squares theory operates under the assumption of fixed covariates. The classical unbiased estimator for the overall conditional average treatment effect remains the weighted difference-in-means estimator given in (\ref{eq:dim}). The true variance for this estimator is inflated, as unlike with the sample average treatment effect we no longer condition on the potential outcomes in each block. Nonetheless, we now demonstrate the variance estimators $S^2_1(Q)$ given in (\ref{eq:sq}) and $S^2_2(Q)$ given in (\ref{eq:sq2}) remain conservative estimators in expectation for $\var(\hat{\Delta} \mid \mathcal{C}, \cZ)$.

\begin{proposition}\label{prop:QCATE}
If $Q$ is constant across all elements of $\Omega$:
\begin{align*}\E[S^2_1(Q) \mid \cC, \cZ] - \var(\hat{\Delta}\mid \cC, \cZ)
&=  \frac{1}{B^2}\mathbf{g}^TW(I-H_Q)W\mathbf{g}\geq 0,
\end{align*}
where $\mathbf{g}$ is a vector of length $B$ with $g_i = ({1-h_{Qii}})^{-1/2}\bar{f}_i$. Further, \begin{align*}&\E[S^2_{2}(Q) \mid \cC, \cZ] - \var(\hat{\Delta}\mid \cC, \cZ)\\
&=  \frac{1}{B^2}\sum_{i=1}^Bw_i^2\var(\hat{\tau}_i\mid \cC, \cZ)\sum_{j\neq i}\frac{h_{Qij}^2}{(1-h_{Qjj})^2} + \frac{1}{B^2}\mathbf{f}^TW(I-H_Q)\Psi_Q(I-H_Q)W\mathbf{f} \geq 0
\end{align*}
\end{proposition}
The proof is analogous to that of Propositions \ref{prop:Q1} and \ref{prop:Q2}.  The insights from Theorem 2 similarly extend variance estimation for the conditional average treatment effect: through using regression adjustments on the average \textit{level} of the covariates in a given block results in less conservative variance estimators, with the degree of improvement now dependent on the extent to which the average of the weighted covariates in a given block are able to predict $w_if_i$, the weighted conditional average treatment effect in a block given the covariate values. 

In the case of equal block sizes, if the stratum-level treatment effects are homoskedastic (i.e. $\var(\hat{\tau}_i\mid \cC, \cZ)$ is constant across all blocks), then we are also entitled to an additional variance estimator connected to $HC2$ standard errors. Let $\tilde{\Psi}_Q$ be a diagonal matrix whose $i^{th}$ diagonal element is $\tilde{\Psi}_{Qii} = 1/(1-h_{Qii})$, and define $S^2_3(Q)$ as
\begin{align}\label{eq:sq3}
{S}^2_{3}(Q)&=\frac{1}{B^2}\bhtau^TW(I-H_Q)\tilde{\Psi}_Q(I-H_Q)W\bhtau,
\end{align}

\begin{proposition}\label{prop:QCATE3}
If $Q$ is constant across all elements of $\Omega$, block sizes are equal (such that $W=I$), and $\var(\hat{\tau}_i\mid\cC, \cZ)$ is constant across blocks:
\begin{align*}\E[S^2_3(Q) \mid \cC, \cZ]- \var(\hat{\Delta}\mid \cC, \cZ)=  \frac{1}{B^2}\bar{\mathbf{f}}^T(I-H_Q)\tilde{\Psi}_Q(I-H_Q)\bar{\mathbf{f}} \geq 0.
\end{align*}
\end{proposition} 
The proof is deferred to the appendix. In the general case with across block heteroskedasticity, unequal block sizes, or when conducting inference on the the sample average treatment effect $S^2_{3}(Q)$ need not be conservative in expectation. It does, however, converge in probability to the same limiting value as $S^2_1(Q)$ and $S^2_2(Q)$, indicating that the prospect of anticonservative inference through $S^2_3(Q)$ may only be a realistic concern in small samples.  

These developments demonstrate that the modes of inference presented for the sample average treatment effect in \S 4 yield harmonious extensions to inference on the conditional average treatment effect. That is, hypothesis tests and confidence intervals for the sample average treatment can also be interpreted as hypothesis tests and confidence intervals for the conditional average treatment effect should the practitioner deem the super-population formulation.

\subsection{Population average treatment effect (PATE)}
As an alternative super-population formulation, suppose we now consider the block sizes $\{n_1,...,n_B\}$ as fixed, but the covariates within a given block, $\{\bx_{i1},...,\bx_{in_{i}}\}$ as random. We now consider the pair of potential outcomes $\{r_{1ij}, r_{0ij}\}$ as having arisen through the following model:
\begin{align*}
\bx_{ij}&= \bzeta_i + \varepsilon_{ij}\\
(r_{1ij}, r_{0ij}) \mid \bx_{ij} &= (f_{1i}(\bx_{ij}), f_{0i}(\bx_{ij})) + (\epsilon_{1ij}, \epsilon_{0ij}),
\end{align*} where $\bzeta_i$ are block-specific fixed effects, $\varepsilon_{ij}$ are $iid$ from some mean zero, finite variance distribution $G$, and the $(\epsilon_{1ij}, \epsilon_{0ij})$ are drawn $iid$ from an arbitrary distribution $F$ with mean $(0,0)$ and block-specific variance-covariance matrix $\Sigma_{i\epsilon}$. Within this super-population abstraction, the \textit{population} average treatment effect, or $PATE$, in a finely stratified experiment is defined as.
\begin{align}\label{eq:PATE}
\bar{\Delta}^{(P)} = \sum_{i=1}^B(n_i/N)\int (f_{1i}(\bzeta_i + \varepsilon_{ij}) - f_{0i}(\bzeta_i + \varepsilon_{ij}))dG(\varepsilon_{ij})
\end{align}

The classical weighted difference-in-means estimator $\hat{\Delta}$ remains an unbiased estimator for the population average treatment effect. \citet{ima08} consider this model in a paired experiment with $\bzeta_i = \bzeta_0$ for all $i$. Therein, they demonstrate not only that the average of the paired differences yields an unbiased estimator for the average population average treatment effect, but that the classical variance estimator for the difference-in-means, $S^2_P$, is an \textit{unbiased} estimator for $\var(\hat{\Delta}|\cZ)$ regardless of whether or not the underlying treatment effect is additive. 

It is here that we see the potential incongruity between inferential methods for the sample average treatment effect and for the population average treatment effect appear. The improvements presented herein empower the practitioner to use the average level of the covariates within a given block as a means to improve variance estimation when the sample or conditional average treatment effects are the targets of estimation. If the target is instead the population average treatment effect as formulated in this section, randomness in $\{\bx_{ij}\}$ renders these conclusions inapplicable. As an illustration, consider the expectation of $S^2_1(Q_2)$ within this super-population formulation.

\begin{align*} &\E[S^2_1(Q_2) \mid \cZ] = \E[\E[S^2_1(Q_2) \mid \cZ, \cC]]\\
&= \E[\var(\hat{\Delta} \mid \cZ, \cC)] + \frac{1}{B^2}\E[\mathbf{g}^TW(I-H_{Q_{2}})W\mathbf{g}\mid \cZ]\\
&=\var(\hat{\Delta} \mid \cZ) - \var\left(\frac{1}{N}\sum_{i=1}^B\sum_{j=1}^{n_{i}}(f_{1i} - f_{0i})\mid \cZ\right) + \frac{1}{B^2}\E[\mathbf{g}^TW(I-H_{Q_{2}})W\mathbf{g}\mid \cZ]\\
&=\var(\hat{\Delta} \mid \cZ) - \var\left(\frac{1}{N}\sum_{i=1}^B\sum_{j=1}^{n_{i}}(f_{1i} - f_{0i})\mid \cZ\right) + \frac{1}{B^2}\E[\mathbf{g}^TW(I-H_{Q_{1}})W\mathbf{g}\mid \cZ] \\&-\frac{1}{B^2}\E[\mathbf{g}^TWH_{M}W\mathbf{g}\mid \cZ] \\
&\approx \var(\hat{\Delta} \mid \cZ) + \frac{1}{B^2}\E[\mathbf{g}\mid \cZ]^TW(I-H_{Q_{1}})W\E[\mathbf{g}\mid \cZ] - \frac{1}{B^2}\E[\mathbf{g}^TWH_{M}W\mathbf{g}\mid \cZ],
\end{align*} where the approximation stems from ignoring the division by $\sqrt{1-h_{Q_{2}ii}}$ in \\$g_i = \bar{f}_i/\sqrt{1-h_{Q_{2}ii}}$ in the term $\E[\mathbf{g}^TW(I-H_{Q_{1}})W\mathbf{g}\mid \cZ]$, a safe approximation in large samples. The last line need not be greater than $\var(\hat{\Delta} \mid \cZ)$. For example, in the case where $\bzeta_i = \bzeta_0$ for all $i$, it will approximately equal $\var(\hat{\Delta} \mid \cF)- B^{-2}\E[\mathbf{g}^TWH_{M}W\mathbf{g}\mid \cZ]$, meaning that it provides an underestimate. In short, the implications of this derivation are that effect modification cannot be safely exploited in variance estimation when conducting inference for the $PATE$ as defined in (\ref{eq:PATE}), while it can be exploited for inference on the $SATE$ and $CATE$. Valid inference for the sample average treatment effect using the developments in  \S 4 may, or may not, be anti-conservative for inference on the population average treatment effect depending on whether or not the corresponding variance estimator was constructed using $M$, the average level of the covariates. As a consequence, $S^2_1(Q_{2})$ and $S^2_2(Q_2)$ cannot be relied upon to yield valid inference for the $PATE$.

The above derivation also illustrates that if the covariate means $\bar{X}$ were not present, the issue with potentially anti-conservative variance estimation disappears. That is, $S^2_{1}(Q_{1})$ and $S^2_2(Q_1)$ can safely be utilized when conducting inference on $PATE$ since these estimators do not exploit effect modification. Similar dissonance for variance estimates for population average treatment effects versus sample average treatment effects is also observed when conducting inference after regression adjustment in completely randomized experiments; compare, for example the suggested variance estimator of \citet{pit13} and \citet{ber13} to that of \citet{lin13}. 

\section{An exact test for additivity with power under linear effect modification}

The developments of the previous sections naturally lend themselves to a new test of the null hypothesis of an additive treatment effect model when the researcher suspects the presence of effect modification on the basis of observed covariates. Suppose we want to test the null hypothesis of an additive treatment effect model against the alternative that there is effect heterogeneity,
\begin{align*}
\mathbf{H}_o: \tau_{ij} &= \bar{\Delta} \;\; \text{for some }\bar{\Delta}, \text{ for all}\;\; i = 1,..,B; j = 1,..,n_i\\
\mathbf{H}_{a}: \tau_{ij} &\neq \tau_{i'j'}.\;\; \text{for some}\;\; i,i', j, j'\\
\end{align*}


Let $F(\bZ)$ be the $F$-ratio for a partial $F$-test comparing a regression of $W\bhtau$ on $Q_1$ to one on $Q_2 = [Q_1,M] = [Q_1, (I-H_{Q_{1}})W\bar{X}]$ with observed treatment allocation $\bZ$,
\begin{align*}
F(\bZ) &= \left(\frac{\bhtau^TW(I-H_{Q_{1}})W\bhtau - \bhtau^TW(I-H_{Q_{2}})W\bhtau}{\bhtau^TW(I-H_{Q_{2}})W\bhtau}\right)\frac{B-rank(Q2)}{K}\\
&= \left(\frac{\bhtau^TWH_{M}W\bhtau}{\bhtau^TW(I-H_{Q_{2}})W\bhtau}\right)\frac{B-rank(Q2)}{K},
\end{align*} where the second line stems from orthogonality of $M$ and $Q_1$. Small values for this ratio indicate that the reduction in residual variation from using $Q_2$ was modest relative to the model only containing $Q_1$. Large values for this ratio indicate substantial reduction in residual variation from exploiting effect modification through $Q_2$. 

Note that while the null hypothesis specifies that the treatment effect is additive, it does not specify the value of the additive treatment effect. That is, in general the true value of the additive effect, call it $\bar{\Delta}_0$, is a nuisance parameter for the desired inference. Fortunately, our choice of test statistic eschews this dependence.

\begin{proposition}
$F(\bZ)$ is a pivotal statistic for testing the null of an additive treatment effect in a finely stratified experiment.
\end{proposition}
\begin{proof}
Suppose the null hypothesis was true and that the additive treatment effect equaled $\bar{\Delta}_0$. Note then that in the $i^{th}$ block, $w_i\hat{\tau}_i$ can be written as 
\begin{align*}
w_i\hat{\tau}_i 
&= w_i\sum_{j=1}^{n_i}\left(Z_{ij}(r_{0ij} + \bar{\Delta}_0)/n_{1i} - (1-Z_{ij})r_{0ij}/n_{0i}\right)\\
&= w_i\bar{\Delta}_0 + w_i\sum_{j=1}^{n_{i}}\left(Z_{ij}r_{0ij}/n_{1i} - (1-Z_{ij})r_{0ij}/n_{0i}\right),
\end{align*}
Hence, the vector $W\bhtau_i$ can be broken into the sum of two vectors, one of which is a mean zero random variable, and the other being the deterministic vector $\bar{\Delta}_0W\bone$. To complete the proof, simply note that $W\bone$ is in the columnspace of both $Q_1$ and $Q_2$, such that the term $\bar{\Delta}_0W\bone$ drops out of both the numerator and denominator of $F(\bZ)$.
\end{proof}

Let $t$ be the observed value of $F(\bZ)$ in the sample at hand. To compute a $p$-value corresponding to $t$, we can simply choose an arbitrary value for the additive treatment effect, say $\bar{\Delta}_0 = 0$, and compute the randomization distribution of $F(\bZ)$ which is entirely specified under the null,
\begin{align}\label{eq:pval}
p_{val} &= \frac{1}{|\Omega|}\sum_{\bz \in \Omega} \chi\left\{F(\bz) \geq t \mid \cF, \cZ, \tau_{ij} = 0\;\; \forall\;\; i,j \right\},
\end{align} where $\chi\{A\}$ is an indicator that the event $A$ occurred. 


Among alternatives to strict additivity, the test will be more powerful when there exists effect modification that is well modeled by a regression on $Q_2$. The test will not be particularly powerful when there exists heterogeneity that is not well modeled as a linear function of the covariates which compose $Q_2$. Regardless, the test will maintain the desired size even in finite samples as, for each fixed value of $\bar{\Delta}_0$, the distribution of $T(\bZ)$ can be computed exactly under the null of additivity.

Note that, in principle, other estimation procedures beyond linear regression could be used to compare sums of squared errors including and not including the observed covariates. In general, the corresponding test statistics will not be pivotal, meaning that their distribution could depend on the value of the additive treatment effect. This can be accommodated through the technique of \citet{ber94} by first finding $1-\gamma$ confidence interval for the value of $\bar{\Delta}$ through inversion of randomization tests under the assumption of additivity \citep{obs}, and finding the maximal $p$-value for the test statistic for values of $\bar{\Delta}_0$ within the confidence interval, and adding $\gamma$ to the result. See \citet{din15mod} for a recent application of this idea to testing for effect variation in completely randomized experiments. 

\section{Illustrations and simulations}
\subsection{Variance estimation in finely stratified experiments}\label{sec:sim1}
We now explore the improvements in inference for the $SATE$ and $CATE$ that can be attained by exploiting effect modification, and illustrate our caveat about these benefits not extending to inference on the $PATE$. There are $B$ blocks, $0.4B$ of which are triplets and  $0.6B$ of which are pairs. The blocks are formed by taking an $iid$ sample of a $k=10$ dimensional vector of covariates $\bx_i$, where each component is $iid$ uniform on the interval [0,1]. Modifying the function utilized in the simulation study of \citet{fri91} to remove linear terms, for each block-level covariate vector $\bx_i$ we then sample potential outcomes under treatment and control from the following distribution:

\begin{align} \label{eq:gen}
r_{1ij} &= a\left(10\sin(\pi x_{i1}x_{i2}) + 20(x_{i3}-1/2)^2 + 10\exp(x_{i4}) +5(x_{i5}-1/2)^3\right) + b\epsilon_{ij}\\
r_{0ij} &= 10\sin(\pi x_{i1}x_{i2}) + 20(x_{i3}-1/2)^2 + 10\exp(x_{i4}) + 5(x_{i5}-1/2)^3 + \epsilon_{ij}\nonumber\\
\epsilon_{ij}&\overset{iid}{\sim}\mathcal{N}(0,1)\nonumber
\end{align}

For the simulations in this subsection, we set $B=100$ for the number of strata, and $a=2$ and $b=2$ in (\ref{eq:gen}). Under this specification the average treatment effect at the population level, $\bar{\Delta}^{(P)}$, is roughly 24.1, and is fixed as an estimand across samples. Further, $\var(\hat{\Delta}\mid \cZ = 0.437$. The sample average treatment effect, $\bar{\Delta}$, and the conditional average treatment effect, $\bar{\Delta}^{(C)}$, vary with each realization, as their definitions depend on $\cF$ and $\cC$ respectively. There is effect heterogeneity present, as $\E[\tau_{ij}\mid \cC]  = 10\sin(\pi x_{i1}x_{i2}) + 20(x_{i3}-1/2)^2 + 10\exp(x_{i4}) + 5(x_{i5}-1/2)^3$. In this generative model, $\E[\var(\hat{\Delta}\mid \cC, \cZ)] = 0.443$, and  $\E[\var(\hat{\Delta}\mid \cF, \cZ)] = 0.401$. In each simulation, we
\begin{enumerate}
\item Simulate covariates $\bx_i$, $i=1,...,100$ and potential responses $(r_{1ij}, r_{0ij})$, $j = 1,..,n_i$, setting $n_i=3$ for 40 blocks and $n_i=2$ for 60 blocks
\item Compute the variance estimators $S^2_1(Q)$, $S^2_2(Q)$, and $S^2_3(Q)$
\end{enumerate}

We form the matrix $Q$ used to compute the variance estimators in three ways,
\begin{enumerate}
\item \textit{None}. Only including a constant column and the stratum weights .
\item \textit{Correct}. Including a constant column, stratum weights, and weighted transformed covariates $w_i\sin(\pi x_{i1}x_{i2})$, $w_i(x_{i3}-1/2)^2$, $w_i\exp(x_{i4})$, $w_i(x_{i5}-1/2)^3$  (\textit{Correct}).
\item \textit{Incorrect}. Including a constant column, stratum weights, and weighted values for the original 10 covariates (without transformation), $w_ix_{ik}$, $k=1,...,10$.
\end{enumerate}
The functional form for effect modification is thus correctly specified within the second form, and incorrectly specified in the third form.

\begin{table}
\caption{\label{tab:7.1}Expectations for variance estimators for various matrices $Q$. Target expectations for valid inference on the $SATE$, $CATE$, and $PATE$ are 0.0401, 0.0443, and 0.437 respectively.}
\centering
\fbox{\begin{tabular}{r c c c}
&\multicolumn{3}{c}{Covariates in $Q$}\\
& {None} & {Correct} & {Incorrect}\\
\hline\hline
$S^2_1(Q)$& 0.437 & 0.0460 &0.108\\
$S^2_2(Q)$& 0.447& 0.0474 & 0.126\\
$S^2_3(Q)$&0.437& 0.0443& 0.110\\
\end{tabular}}
\end{table}
Table \ref{tab:7.1} shows the results of this simulation. We see that, as Propositions 1-4 guarantee, $S^2_1(\cdot)$ and $S^2_2(\cdot)$ remained conservative in expectation for the variances for estimating the $SATE$ and $CATE$ for all choices of $Q$. Using the correctly specified covariates allows the expectations to come closest to the true values for the variances, while the incorrect specification still performs substantially better than ignoring the covariates altogether. For inference on the $PATE$, we see that only the choice of $Q$ which ignores the covariates yields a valid estimator for the variance; the choices incorporating the covariates would produce substantially anticonservative inference for the population average treatment effect. While not guaranteed to be as such in this simulation, we see that $S^2_3(\cdot)$ produced estimators which remained conservative in expectation for inference on the $SATE$ and $CATE$ for all three choices of $Q$, and for the $PATE$ through the choice of $Q$ ignoring the covariates.
\subsection{Testing for effect heterogeneity}
We now demonstrate the test for effect modification proposed in \S 6. We use a similar generative model for the covariates as was employed in \S \ref{sec:sim1}, but we instead set $b=1$ and conduct the test for effect heterogeneity for increasing values of $a$ in (\ref{eq:gen}). At $a=1$, the null hypothesis of additivity is true; all values $a\neq 1$ imply that the null is false. We also set $B=20$ as a means of illustrating the exactness of the test. As in the previous section, we conduct the test utilizing both the correct and incorrect specification for the functional form of the covariates in forming the matrix $Q_2$. Hence, in each iteration we 
\begin{enumerate}
\item Simulate covariates $\bx_i$, $i=1,...,20$ and potential responses $(r_{1ij}, r_{0ij})$, $j = 1,..,n_i$, setting $n_i=3$ for 8 blocks and $n_i=2$ for 12 blocks
\item Randomly allocate individuals to treatment or control in accordance with the finely stratified design, recording the observed outcomes and the value for the test statistic $F(\bZ)$
\item Estimate the permutation $p$-value in (\ref{eq:pval}) through Monte Carlo simulation.
\end{enumerate}

\begin{figure}[h]
\centering
\makebox{\includegraphics[scale=.5]{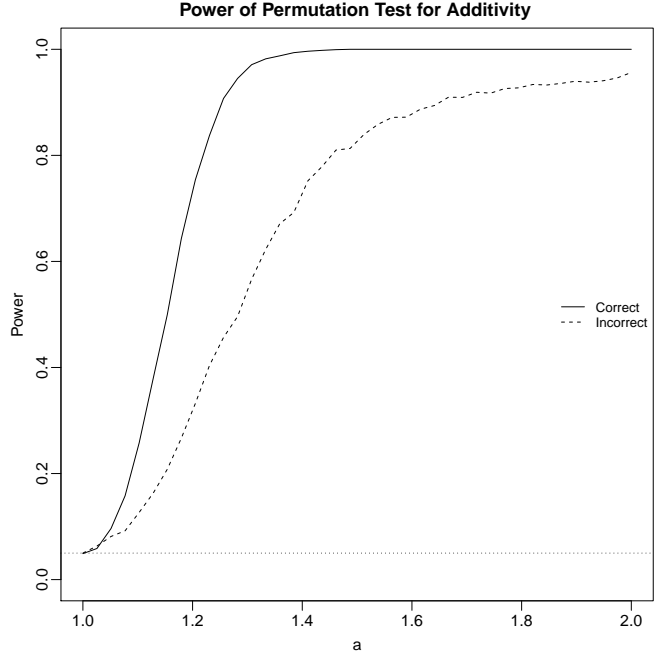}}
\caption{\label{fig:power}The power of the permutation test for effect modification under both correct and incorrect covariate specification as a function of $a$, which controls the departure from additivity. $a=1$ corresponds to the null being true. The horizontal line corresponds to $\alpha=0.05$, the desired size in the simulation.}
\end{figure}

We set $\alpha=0.05$ for this study. Figure \ref{fig:power} shows the power of our testing procedure as a function of $a$, under both correct and incorrect specifications for $Q_2$. Note that at $a=0$ our tests have the correct size. As $a$ increases, the power increases for both choices of $Q_2$, but more rapidly for the correct choice of $Q_2$. Hence, while the specification of $Q_2$ affects the power of the test, it does not affect its validity in terms of maintaining the desired Type I error rate. 
\newpage
\subsection{Pairs or quartets?}\label{sec:2v4}
In this study, we explore the extent to which the analytical limitations of paired experiments described in the introduction are mitigated by the methods presented herein. In our simulation study, we have $N=80$ total individuals. We have a single covariate for each individual. In our study, we fix the observed values for the covariate at $\mathbf{x} = (0.25,0.25, 0.5,.0.5,...,10, 10)$. Hence, there are 40 pairs of individuals who share the same value for the covariate, hence forming natural pairs. Due to concerns over the analytical limitations of paired experiments, the practitioner may instead choose to create 20 quartets of individuals, namely those taking on values $\{0.25, 0.25, 0.5, 0.5\}$, $\{0.75,0.75,1,1\}$,....$\{9.75,9.75,10,10\}$. 

For each experiment, we generate potential outcomes for individual $j$ as
\begin{align*}
r_{1j} = 100 + 30x_j + \epsilon_j,\;\;\;
r_{0j} = 20x_j + \epsilon_j,\;\;\;
\epsilon_j \overset{iid}{\sim} \mathcal{N}(0, 10^2).
\end{align*}

$x_j$ is thus an effect modifier, as the treatment effect for individual $j$ is $100 + 10x_j$. We first consider two situations: one in which the practitioner has access to $\mathbf{x}$ itself, and one in which the practitioner instead has access to $b_{j} = \exp({x}_j/3)$ for each individual. Let $Q_x$ be the $B\times 2$ matrix with $\bone$ in the first column and the value of $x_j$ defining each pair in the second column, and let $Q_b$ be the analogous for the incorrectly specified covariate $b_j$ . We imagine the target of inference is the conditional average treatment effect, and hence seek to estimate $\var(\hat{\Delta}\mid\cC, \cZ)$ for both the paired and quartet design. Under this specification, we simply utilize the bias formulae in Propositions 1, 2 and 4 to calculate the expectation for the variance estimators $S^2_P$, $S^2_\ell(Q_x)$ and $S^2_\ell(Q_b)$ for $\ell=1,2,3$. For the quartet design, we simply compute the true value of the variance, along with the bias of the conventional estimator $S^2_{CS}$. 

Table \ref{tab:2v4} contains relevant numerical information for comparing estimation under the two designs. We first see that the true variance under the paired design is smaller than that under the quartet design (5.00 vs 5.65), such that if we had access to this true variance the paired design would undoubtedly be preferred; however, the classical variance estimator in a paired experiment has an expectation of 21.35, while the classical variance estimator in the quartet experiment has an expectation of 5.68. This comparison of conventional variance estimators highlights the motivation for recommending quartet experiments over paired experiments within the literature. In the rows labeled ``Correct, Linear" and ``Incorrect, Linear" we see the improvements in variance estimation both under proper specification and misspecification of the response function. When the response function is properly specified both $S^2_1(Q_x)$ and $S^2_2(Q_{x})$ are lower in expectation than $S^2_{CS}$, showing that proper modeling of the treatment effect heterogeneity provides variance estimators whose expectations are smaller than that of the coarsely stratified experiments. When the heterogeneity is not properly modeled, improvements in the variance estimator are still attained; however, the expectations for the variance estimators now exceed that from the quartet experiment.

\begin{table}
\caption{\label{tab:2v4}A comparison of variance estimation in the paired and quartet designs}
\centering
\fbox{\begin{tabular}{c c c c c c|c c}
\multicolumn{6}{c}{Pairs} & \multicolumn{2}{|c}{Quartets}\\
Covariates &$\var(\widehat{\Delta})$ &$S^2_P$ & $S^2_1$ & $S^2_2$ & $S^2_{3}$& $\var(\widehat{\Delta})$ & $S^2_{CS}$\\
\hline
&{5.00}&21.35&-&-&-&{5.65}&{5.68}\\
Correct, Linear & - & -&5.09 &5.26 &5.00 &-&-\\
Incorrect, Linear&-&-&7.12&8.76&8.25&-&-\\
Correct, Cubic&-&-&5.52&5.59&5.00&-&-\\
Incorrect, Cubic&-&-&7.23& 6.06&5.24&-&-\\
\end{tabular}}
\end{table}

Recall once again that the findings of Propositions 1 and 2 facilitate conservative estimation of the variance for any fixed matrix $Q$. This allows us, before conducting the experiment, to decide to include polynomial terms in the matrix $Q$ to more flexibly model the relationship between the covariates at hand and the treatment effects. Suppose we now add quadratic and cubic terms of $\mathbf{x}$ and $\mathbf{b}$, calling the corresponding matrices $Q_{x3}$ and $Q_{b3}$ respectively. In the last two rows of \ref{tab:2v4}, we consider the performance of $S^2_1(Q_{x3}), S^2_1(Q_{b3})$, $S^2_2(Q_{x3})$, and $S^2_2(Q_{b3})$. We see that both $S^2_1(Q_{x3})$ and $S^2_2(Q_{x3})$ have a larger expectation than what was attained when we omitted the polynomial terms. By adding two more predictor variables, the sum of diagonals of the correspond hat matrix increases from 2 to 4, hence resulting in additional inflation of residuals. For $S^2_{1}(Q_{b3})$, we see that this inflation has also swamped any benefit from flexibility in modeling as its expectation is larger than that of $S^2_1(Q_b)$. For $S^2_2(Q_{b3})$, we see that the additional flexibility has been beneficial, and the expectation for the variance estimator has decreased relative to $S^2_2(Q_{b3})$, although not enough to fall below the level of $S^2_{CS}$.

A component of the remaining conservativeness of the estimators $S^2_{1}(Q_{b3})$ and $S^2_2(Q_{b3})$ stems from our variance estimators being unbiased in expectation regardless of the degree of heteroskedasticity across blocks, and hence having to be inflated in the presence of high leverage points. If one is willing to do away with the requirement, $S^2_{3}(Q_{b3})$ becomes an appealing estimator. This estimator combines elements of $S^2_{1}(Q)$, essentially adjusting the variance estimator by $1/(1-h_{Qii})$ instead of $1/(1-h_{Qii})^2$, and $S^2_2(Q)$, by adjusting residuals instead of responses to account for influential points. The column labeled $S^2_{3}(\cdot)$ corresponds to this estimator. In the setting considered herein, it is exactly unbiased with $Q_{x}$ and $Q_{b3}$, owing to the fact that $\var(\hat{\tau}_i\mid \cC, \cZ)$ is constant across pairs. It is necessarily less conservative than $S^2_2(\cdot)$ for all four choices of $Q$ considered, and it is less conservative than $S^2_1(\cdot)$ for all choices of $Q$ except for $Q_b$. As $B$ decreases the estimators all converge to the same limit, yet here we see the potential benefits of using the estimator $S^2_{3}(\cdot)$.

\section{An example: The Children's Television Workshop Experiment}

\citet{bal73} designed an experiment to evaluate an educational television program which sought to improve reading skills for young children. \S 10.7 of \citet{imb15} examined a subset of the experiment conducted in Youngstown, Ohio with $B=8$ primary schools. In each school, a pair of first-grade classes was selected, with one class in each pair assigned to watch the show during reading class and the other class assigned to continuing with the usual curriculum. Each class has a pre-test score assessing average reading ability, $x_{ij}$ in our notation, along with a post-test after the experiment, $R_{ij} = Z_{ij}r_{1ij} + (1-Z_{ij})r_{0ij}$, where $Z_{ij}$ is 1 if the class was shown the educational program \textit{The Electric Company}. $\hat{\tau}_i$ is the difference between the observed treatment and control scores on the post-test in the $i$th pair.

In this data set, the conventional difference-in-means estimator was $\hat{\tau} = 13.4$, with an observed value of the conventional standard error of $S_P = 4.6$. We now consider using the estimators developed herein to improve upon this standard error estimate. We include linear and quadratic terms in the covariates, defining $Q_2 = [\bone, (I - \bone\bone^T/B)\bar{X}]$, where the $\{i,1\}$ entry of  $\bar{X}$ is $\bar{x}_{i1} = (x_{i1}+x_{i2})/2$, and the $\{i,2\}$ entry of $\bar{X}$ is $\bar{x}_{i2} = (x_{i1}^2+x_{i2}^2)/2$. The values for $S_1(Q_2)$, $S_2(Q_2)$, and $S_3(Q_2)$ are 4.2, 4.34, and 3.57 respectively. All three estimators would thus facilitate the construction of narrower confidence intervals than the ones constructed using $S_P$ while maintaining the conclusion that the treatment was effective at $\alpha=0.05$. As noted, $S^2_3(Q_2)$ is not in general unbiased for the variance when the target of inference is the sample average treatment effect, so the discrepancy between this estimator and the other two may well stem from downwards bias. This concern is not relevant for the other two estimators, a reason to prefer them particularly in small samples.

\section{Discussion}
When the target of estimation is either the sample or the conditional average treatment effect, the developments presented in this work facilitate improved variance estimation for finely stratified experiments for inference conducted based on both the conventional difference-in-means estimator and estimators utilizing regression adjustment. As the simulation study in \S\ref{sec:2v4} illustrated, these have certainly mitigated the analytical limitations of finely stratified experiments by providing more powerful inference than that available through classical variance estimators, yet the analytical issues have not been entirely resolved. If the regression model is grossly misspecified, the variance estimators presented herein may not provide an improvement over that of an experiment with blocks of size four. 

One direction for future research is investigating the extent to which the improvements presented in this work can be employed in the super-population setting considered by \citet{van12} wherein rather than pairs being drawn $iid$, individuals are drawn $iid$ and then optimally paired after being selected into the study. More generally, the nature of these improvements raise additional questions about the extent to which inference on local estimands in randomized experiments should be transferable to population-level estimands in popular super-population formulations. With respect to the conventional variance estimator in a completely randomized experiment \citet{imb15} describe the consonance between finite-population and super-population inference through this variance estimator as an ``attractive property." \citep[\S 6.7, p.101]{imb15}. Yet as was noted in \S 6.3, variance estimators which exploit effect heterogeneity can yield anticonservative inference at the level of the $PATE$ as defined in \citet{ima08}. Our perspective is that rather than detracting from the appeal of these new estimators, this dissonance forces the researcher to critically assess the question, ``to whom does the inference apply?" The answer is often left ambiguous in the analysis of randomized experiments, and $iid$ assumptions are often made vacuously, without consideration of the true nature of the process by which the data came to be and the corresponding ramifications for the integrity of the performed inference. 

\appendix
\section{Lemmas}
\begin{lemma} \label{lemma:1} Under Conditions 1-3, $B^{-1}\sum_{i=1}^Bw_i\hat{\tau}_im_{ik}$ converges in probability to \\$\underset{n\rightarrow\infty}{\lim}B^{-1}\sum_{i=1}^nw_i\bar{\tau}_im_{ik}$ for any $k=1,...,rank(M)$. Further, and $B^{-1}\sum_{i=1}^Bw_i^2\hat{\tau}_i$ and $B^{-1}\sum_{i=1}^Bw_i\hat{\tau}_i$ converges in probability to $\underset{n\rightarrow\infty}{\lim}B^{-1}\sum_{i=1}^nw_i^2\bar{\tau}_i$ and  $\underset{n\rightarrow\infty}{\lim}B^{-1}\sum_{i=1}^nw_i\bar{\tau}_i$ respectively.
\end{lemma}
\begin{proof}
We prove the result for $B^{-1}\sum_{i=1}^Bw_i\hat{\tau}_im_{ik}$; the proof for the remaining two weighted sums are analogous. For any $k$, $\E[B^{-1}\sum_{i=1}^Bw_i\hat{\tau}_im_{ik}\mid\cF, \cZ] = B^{-1}\sum_{i=1}^Bw_i\bar{\tau}_im_{ik}$, which has a finite limit by Condition 3. We now show that $\var(B^{-1}\sum_{i=1}^nw_i\hat{\tau}_im_{ik}\mid\cF, \cZ)$ converges to zero.
\begin{align*}\var\left(B^{-1}\sum_{i=1}^nw_i\hat{\tau}_im_{ik}\mid\cF, \cZ\right)&=B^{-2}\sum_{i=1}^Bw_i^2\left({\sigma^2_{1i}}/{n_{1i}} + {\sigma^2_{0i}}/({n_{0i}}) - {\sigma^2_{\tau i}}/{n_i}\right)(m_{ik})^2\\ 
&\leq B^{-2}\left\{\sum_{i=1}^B\left(\sum_{j=1}^{n_{i}}w_i^2(r_{1ij}^2+r_{0ij}^2)\right)^2\right\}^{1/2}\left\{\sum_{i=1}^Bm_{ik}^4\right\}^{1/2}\\
& \leq B^{-2}\left\{\sum_{i=1}^Bn_i^2\left(\sum_{j=1}^{n_{i}}w_i^4(r_{1ij}^4/n_i+r_{0ij}^4/n_i)\right)\right\}^{1/2}\left\{\sum_{i=1}^Bm_{ik}^4\right\}^{1/2}\\
& \leq C_1C_2/B\\
\end{align*}
by Conditions \ref{cond:1} and \ref{cond:2}, which tends to zero as $B\rightarrow \infty$. Chebyshev's inequality and Condition \ref{cond:2} complete the proof. 
\end{proof}
\begin{lemma}Under Condition 2, $h_{ii} \rightarrow 0$.
\end{lemma}
\begin{proof} From Condition 2, we have that $B^{-1}M^TM$ converges to a finite, invertible matrix; let $\Lambda = (\lim_{B\rightarrow\infty}B^{-1}M^TM)^{-1}$. Note that $h_{ii}$ = $m_i^T(M^TM)^{-1}\bm_i = B^{-1}(m_i)^T(B^{-1}M^TM)^{-1}(\bm_i)$
\begin{align*}
\underset{B \rightarrow \infty}{\lim}h_{ii} &= \underset{B \rightarrow \infty}{\lim} B^{-1}\bm_i^T\Lambda\bm_i = 0
\end{align*}
\end{proof}
\begin{lemma} Under Conditions 1-3 and conditional on $\cF, \cZ$, 
\begin{align*} B^{-1}\sum_{i=1}^Bw_i^2\hat{\tau}^2_{i} \overset{p}{\rightarrow} \lim_{B\rightarrow\infty}B^{-1}\sum_{i=1}^Bw_i^2\left(\bar{\tau}_i^2 +{\sigma^2_{1i}}/{n_{1i}} + {\sigma^2_{0i}}/({n_{0i}}) - {\sigma^2_{\tau i}}/{n_i}\right), \end{align*} 

\end{lemma}
\begin{proof} We have that $\E[B^{-1}\sum_{i=1}^Bw_i^2\hat{\tau}_i^2 \mid \cF, \cZ] = B^{-1}\sum_{i=1}^Bw_i^2(\bar{\tau}_i^2 +{\sigma^2_{1i}}/{n_{1i}} + {\sigma^2_{0i}}/({n_{0i}})$ $- {\sigma^2_{\tau i}}/{n_i})$. It suffices to show that $\var(B^{-1}\sum_{i=1}^Bw_i^2\hat{\tau}_i^2\mid\cF, \cZ)$ converges to zero.

\begin{align*}
&\var\left(B^{-1}\sum_{i=1}^Bw_i^2\hat{\tau}_i^2\mid\cF, \cZ\right) \\&= B^{-2}\sum_{i=1}^B\var(w_i^2\hat{\tau}_i^2\mid\cF, \cZ) \leq B^{-2}\sum_{i=1}^B\E[w_i^4\hat{\tau}_i^4\mid\cF, \cZ] \\
&\leq B^{-2}\sum_{i=1}^Bn_i^2w_i^4\E\left[\left(\sum_{j=1}^{n_i}(Z_{ij}r_{1ij'}^2/n_{1i}^2 - (1-Z_{ij})r_{0ij}^2/n_{0i}^2)\right)^2\right]\\
&\leq B^{-2}\sum_{i=1}^Bn_i^2w_i^4\E\left[\left(\sum_{j=1}^{n_i}Z_{ij}r_{1ij}^2/n_{1i}^2\right)^2+ \left(\sum_{j=1}^{n_i}(1-Z_{ij})r_{0ij}^2/n_{0i}^2\right)^2\right]\\
&\leq B^{-2}\sum_{i=1}^Bn_i^3w_i^4\sum_{j=1}^{n_i}r_{1ij}^4/n_{1i}^4 +  B^{-2}\sum_{i=1}^Bn_i^3w_i^4\sum_{j=1}^{n_i}r_{0ij}^4/n_{0i}^4\\
&\leq 2C_1^4C_2/B,
\end{align*}
which tends to zero as $B\rightarrow \infty$. 
\end{proof}

\section{Proof of Theorem 1}
Noting that the random variables $w_i\hat{\tau}_i$ are independent, it suffices to show that the triangular array version of Lyapunov's condition is satisfied. Let $s^2_B = \sum_{i=1}^B\var(w_i\hat{\tau}_i\mid\cF,\cZ)$. As was demonstrated in the proof of Lemma 3, $B^{-1}\sum_{i=1}^B\E[w_i^4\hat{\tau}_i^4\mid\cF, \cZ]\leq C^*_1$ for a constant $C^*_1$. By Condition 2, $B^{-1}\sum_{i=1}^Bw_i^4\bar{\tau}_i^4 \leq C_2$ for a constant $C_2$. Further, by Condition 3 we have that $s^2_B/B = B^{-1}\sum_{i=1}^Bw_i^2\left({\sigma^2_{1i}}/{n_{1i}} + {\sigma^2_{0i}}/({n_{0i}}) - {\sigma^2_{\tau i}}/{n_i}\right)$ has a finite limit as $B\rightarrow\infty$, call it $L^*$. Hence, using a standard moment inequality,

\begin{align*}\underset{B\rightarrow\infty}{\lim}\frac{1}{s_B^4}\sum_{i=1}^B\E[|w_i\hat{\tau}_i - w_i\bar{\tau}_i |^4 \mid \cF, \cZ]
&\leq \underset{B\rightarrow\infty}{\lim}\frac{8}{B^2(s^2_B/B)^2}B\sum_{i=1}^B\E[w_i^4\hat{\tau}_i^4 \mid \cF, \cZ]/B + w_i^4\bar{\tau}_i^4/B\\ 
&\leq \underset{B\rightarrow\infty}{\lim} \frac{8}{B^2L^*}{B(C_1^* + C_2)}=0. 
\end{align*}
The conditions for Lyapunov's Central Limit Theorem are thus satisfied for $s_B^{-1}\left(\sum_{i=1}^Bw_i(\hat{\tau}_i - \bar{\tau}_i)\right)$.

\section{Proof of Theorem 2}
We prove the result for $S^2_1(Q_2)$ in the case of unequal block sizes. Let $\boeta_{Q_{1}} = [\bar{\Delta}, B^{-1}\lim_{B\rightarrow\infty}(w_i-1)w_i\hat{\tau}_i]$. Let $\Sigma_{Q_{1}}$ be a $2\times 2$ diagonal matrix with $\Sigma_{Q_{1}11}=1$ and $\Sigma_{Q_{1}22} = \lim_{B\rightarrow\infty}B^{-1}\sum_{i=1}^B(w_i-1)^2$. Let $\bbeta_{Q_{1}} = \Sigma_{Q_{1}}^{-1}\boeta_{Q_{1}}$.  Recalling that $Q_1$ and $M$ are orthogonal, we decompose $BS^2_1(Q_{2})$ as
\begin{align*}
BS^2_2(Q_{2}) &= B^{-1}(y^TW(I-H_{Q_{2}})Wy)\\
&= B^{-1}\left(y^TWWy - y^TWQ_{1}(Q_{1}^TQ_{1})^{-1}Q_{1}^TWy - y^TWM(M^TM)^{-1}M^TWy\right)
\end{align*} 
$B^{-1}y^TWWy$ converges in probability to $\lim B^{-1}(\sum_{i=1}^Bw_i^2(\bar{\tau}_i^2 + {\sigma^2_{1i}}/{n_{1i}} + {\sigma^2_{0i}}/({n_{0i}})- {\sigma^2_{\tau i}}/{n_i}))$ by Lemmas 1 and 3. By Lemmas 2 and 3, $B^{-1}y^TWQ_{1}(Q_{1}^TQ_{1})^{-1}Q_{1}^TWy$ converges in probability to $\bbeta_{Q_{1}}^T\Sigma_{Q_{1}}\bbeta_{Q_{1}}$, and $B^{-1}y^TWM(M^TM)^{-1}M^TWy$ converges in probability to $\bbeta_{M}^T\Sigma_M\bbeta_M$. Hence,
\begin{align*}B\left(S^2_{1}(Q_2) - \var(\hat{\Delta}\mid \cF, \cZ)\right) &\overset{p}{\rightarrow} B^{-1}\bbtau^TW(I-H_{Q_{2}})W\bbtau\end{align*} as desired. The proof for $S^2_1(Q_1)$ simply follows by eliminating the terms above pertaining to the matrix $M$. The proofs for $S^2_2(Q_1)$ and $S^2_2(Q_2)$ are analogous.
\section{Proof of Proposition 4}
\begin{proof} Define $\bar{\mathbf{f}}$ as before, and let $\tilde{\Sigma}$ be the covariance matrix for $\bhtau\mid \cC$, which by assumption homoskedasticity and equal block sizes has constant diagonal elements, call them $\nu$. 
\begin{align*}
B^2E[{S}^2_3(Q)\mid \cF, \cZ] &= tr(\tilde{\Sigma} (I-H_Q)\tilde{\Psi}_Q(I-H_Q)) + \bar{\mathbf{f}}^T(I-H_Q)\tilde{\Psi}_Q(I-H_Q)\bar{\mathbf{f}}.\end{align*}
The trace of $\tilde{\Sigma} (I-H_Q)\tilde{\Psi}_Q(I-H_Q)$ is given by
\begin{align*}
tr(\tilde{\Sigma} W(I-H_Q)\tilde{\Psi}_Q(I-H_Q)W) &=\nu \sum_{i=1}^B\left((1 - h_{Qii}) +  \sum_{j\neq i}\frac{h_{Qij}^2}{1-h_{Qjj}}\right)\\
&= \nu \sum_{i=1}^B\left((1 - h_{Qii}) +  \sum_{j\neq i}\frac{h_{Qij}^2}{1-h_{Qii}}\right)\\
&= \nu \sum_{i=1}^B\left(1-h_{Qii} +  h_{Qii}\right) = B\nu\\
\end{align*}
The second line utilizes symmetry of $I-H_{Q}$, while the third utilizes idempotence of $H_Q$, implying that $\sum_{j\neq i}h_{Qij}^2 = h_{Qii}(1-h_{Qii})$. Noting that $\var(\hat{\Delta}\mid \cF, \cZ) = \nu/B$ under the assumptions of the proposition and that $(I-H_Q)\tilde{\Psi}_Q(I-H_Q)$ is positive semidefinite completes the proof.
\end{proof}

\bibliographystyle{apalike}
\bibliography{../bibliography.bib}

\end{document}